\documentclass[11pt]{article}
\usepackage{amsmath,amsfonts,amsthm,amssymb}
\usepackage{fontenc,indentfirst, delarray}
\usepackage{graphicx}
\usepackage{xcolor,fontenc,indentfirst,delarray,graphicx,graphics,epstopdf}
\usepackage{pdfpages,lastpage,fancyhdr,textcomp}
\DeclareGraphicsExtensions{eps,ps,p.gz,pdf}
\usepackage[pdfstartview=FitH]{hyperref}
\usepackage{etaremune}

\makeatletter
\renewcommand{\section}{\@startsection {section}{1}{\z@}%
             {-3.5ex \@plus -1ex \@minus -.2ex}%
             {2.3ex \@plus.2ex}
 {\normalfont\Large\bfseries}}
\renewcommand{\subsection}{\@startsection {subsection}{1}{\z@}%
             {-3.5ex \@plus -1ex \@minus -.2ex}%
             {2.3ex \@plus.2ex}
 {\normalfont\normalsize\bfseries}}
\renewcommand{\subsubsection}{\@startsection {subsubsection}{1}{\z@}%
             {-3.5ex \@plus -1ex \@minus -.2ex}%
             {2.3ex \@plus.2ex}
 {\normalfont\normalsize\bfseries}}
\makeatother

\definecolor{bleuvert}{rgb}{.1,.5,.4}
\definecolor{bleu}{rgb}{.2,.3,.5}
\definecolor{light-gray}{gray}{0.95}
\definecolor{gray}{gray}{0.75}
\definecolor{violet}{rgb}{0.4,0.,0.3}
\definecolor{jaune}{rgb}{0.8,0.6,0.1}
\definecolor{cvert}{rgb}{0.8,0.6,0.5}

\newcommand{\bleu}[1]{{\color{bleu} #1}} 
 
\newtheorem{thm}{Theorem}[section]

\newtheorem{prop}[thm]{Proposition}
\newtheorem{cor}[thm]{Corollary}
\newtheorem{defi}[thm]{Definition}
\newtheorem{exem}[thm]{Example}
\newtheorem{rem}[thm]{Remark}

\newcommand{\B}{\mathcal{B}}

\def\begf{\begin{frame}}
\def\enf{\end{frame}}

\def\begz{\begin{itemize}}

\def\endz{\end{itemize}}

\def\lp{\left(} 
\def\rp{\right)} 
\def\dm{\lp\begin{array}}	
\def\fm{\end{array}\rp}
\def\begf{\begin{frame}}
\def\enf{\end{frame}}

\def\begz{\begin{itemize}}

\def\endz{\end{itemize}}

\def\lp{\left(} 
\def\rp{\right)} 
\def\dm{\lp\begin{array}}	
\def\fm{\end{array}\rp}
\def\m2{M_2 \lp \cc \rp}
\def\m3{M_3 \lp \cc \rp}

\def\ds{\partial\!\!\!\slash}

\def\cc{{\mathbb{C}}}
\def\C{{\mathbb{C}}}	

\def\R{{\mathbb{R}}}

\def\I{{\mathbb{I}}}

\def\mm{{\mathcal M}}	
\def\M{{\mathcal M}}

\def\A{{\mathcal A}}

\def\hh{{\mathcal H}}

\def\K{{\mathcal K}}
\def\BH{{\mathcal B}({\mathcal H})}

\def\cinf{C^{\infty}\lp\mm\rp}

\def\xo0{\omega^0_x}
\def\yo0{\omega^0_y}

\def\xo0{x_\omega^0}
\def\yo0{y_\omega^0}

\def\fm{\Phi(x^\mu)}

\def\dm{\partial_\mu}

\def\dmm{\left(\begin{array}}
\def\fmm{\end{array}\right)}

\newcommand{\HH}{\mathcal{H}}

\usepackage[utf8]{inputenc}
\usepackage[TS1,T1]{fontenc}
\usepackage[top=2.5cm,bottom=2.5cm,left=2.5cm,right=2.5cm,heightrounded,bindingoffset=0mm]{geometry}

\definecolor{grispalebis}{rgb}{.95,.95,.95}

\begin{document}

\title{Twisted Standard Model and its Krein structure\\  {\emph{in memoriam Manuele Filaci}} }\author{P. Martinetti \\[6pt] Dpt. di Matematica, universit\`a di Genova \& INFN}

\maketitle

  \begin{abstract}
 We review the contributions of Manuele Filaci - a PhD student from the
 university of Genova  prematurely deceased a little more than a year
 ago -  to the description of the Standard Model in noncommutative
 geometry. Building on Manuele's discovery that there exist various
 ways to minimally twist the spectral triple of the Standard Model, we
 study in a systematic way the inner product induced by the
 twist. Under loose
 assumptions, this product turns the Hilbert space of the spectral
 triple into a Krein space. For the Standard
 Model, the group of unitary with respect to the twisted product contains the
 symmetry group of twistors as a subgroup.
  \end{abstract}

\section{Introduction}

Manuele Filaci was a  student from the university of Genoa, who
obtained a PhD in 2021.  He started soon after a postdoc at the Jagiellonian 
University in Krakow, but his bright future as a scientist was brutally stopped by a cruel disease. Manuele passed
away in november 2024, just a few months before the beginning of the conference on
\emph{Applications of noncommutative geometry to gauge theories, field
  theories ans quantum spacetimes} at CIRM
(Luminy, France) where
his contribution was awaited with great interest and curiosity.  This
paper is to honor his memory, his work, and
highlight how his contribution to science is rich in future developments.

Manuel started a PhD on neutrinos phenomenology \cite{Biggio:2020aa}, then decided to
switch to mathematical physics. This was a courageous move and we
decided to work together on the applications of
noncommutative geometry to high energy physics, particularly the
description of the Standard Model of fundamental interactions in terms
of spectral triple \cite{Connes:1996fu}. Several important questions
remain open regarding the neutrino sector of
the Standard Model: are neutrinos Dirac or Majorana spinors, what is their mass, do 
right-handed neutrinos exist ? 
In the spectral description of the Standard Model as well, neutrinos
do not behave as the other fermions. They are
transparent to \emph{fluctuation of the metric}, and therefore do not
contribute to the generation of the bosonic sector of the
model. 

Indeed, while the fermionic fields belong to the Hilbert space
$\HH$ that enters the definition of the spectral triple{\footnote{$\A$ is an algebra acting as bounded operators on a
    separable Hilbert space $\HH$
    through a faithful involutive representation $\pi$,  and
    $D$ is an operator on $\HH$ specified in \eqref{eq:specSM}
    below. One usually omits the symbol $\pi$ of representation. We
    will reintroduce it when needed. In all the paper, we assume $\A$
    is a unital algebra with unit $\bf 1$.}} $(\A, \HH, D)$ associated with the Standard Model \cite{Chamseddine:2007oz},
the bosonic fields are part of the (generalised) connection $1$-form 
\begin{equation} \label{eq:7} A + J A J^{-1} \end{equation} 
where $A$ is a selfadjoint element of the set of (generalised)
$1$-forms \begin{equation} \label{eq:8} \Omega^1_D(\A):= \left\{
    \sum_i a_i \, [D, b_i], a_i, b_i\in \A\right\}, \end{equation} and
$J$ is the \emph{real structure} \cite{Connes:1995kx}, namely an
antiunitary operator on $\HH$ such
that \begin{equation} \label{eq:33} J^2=\epsilon \I, \qquad JD =
  \epsilon' DJ,\qquad J\Gamma =\epsilon''\Gamma J \end{equation} where
$\epsilon, \epsilon',\epsilon''$ in $\left\{-1, 1\right\}$ defines the
\emph{KO-dimension}  and $\Gamma$ is the grading, defined below in \S
\ref{subsec:twist-by-grading}.\linebreak   The operator $D$ is parametrised by the physical inputs of the theory, namely the Yukawa couplings of fermions, the Cabibbo matrix and the PNMS mixing matrix for neutrinos. It turns out that the part of $D$ containing the Majorana mass of the neutrinos - and only this part - commutes with the algebra, hence does not contribute to the connection \eqref{eq:7}. 
The substitution of the operator $D$ with the \emph{covariant
  operator} \begin{equation} \label{eq:12} D_A:= D + A +
  JAJ^{-1} \end{equation} is called a \emph{fluctuation of the
  operator $D$} (or fluctuation of the metric, for $D$ defines the
metric through Connes spectral distance formula
\cite{Connes:1992bc}). That is why we summarise the previous
discussion stating that neutrinos are transparent to fluctuations of
the metric.
\bigskip

 This transparency was well identified from the beginning
\cite{Chamseddine:2007oz} and did not seem
particularly problematic, until the discovery of the Higgs boson in
2012.  Connes and Chamseddine then realised  that making the Majorana mass of the neutrino
contributing to the bosonic part of the model would allow to fit the
prediction of the Higgs mass in noncommutative geometry, and would
also  provide the kind of scalar  field required to solve the meta-stability problem of
the electroweak vacuum \cite{Chamseddine:2012fk}. 

A way to turn the
Majorana mass into a scalar field $\sigma$  was proposed  in
\cite{Chamseddine:2013fk}, following  ideas anticipated in \cite[\S
9.1, 9.2]{Connes:2010fk}). It consists  in relaxing the first-order condition (one of the conditions
that defines a spectral triple), fluctuating the operator $D$ then
retrieving the first-order condition dynamically, by a spontaneous breaking of a higher symmetry. This yields a Pati-Salam model of
grand unification~\cite{Chamseddine:2013uq}.

Another proposal was to start with an algebra  bigger than
the one of the Standard Model, so that it  would no
longer commute with the Majorana part of the Dirac operator \cite{Devastato:2013fk}. This  led
to a description of the Standard Model by a \emph{twisted spectral
  triple} \cite{buckley} (see \cite{Filaci:2023aa} for an extensive
review).  Manuele's contribution has been 
decisive: elaborating on a previous model where only part of the
algebra of the Standard Model had been enlarged \cite{buckley}, he worked out
a twisted version of the full Standard Model \cite{Filaci:2021ab}. He
also discovered that contrary to our initial hope, the
first-order condition was  not preserved, not even in its twisted
form. We then started to classify  all the possible ways to twist
the Standard Model, hoping to individuate one of them that would both
preserve the twisted first-order condition and generate the
extra-scalar field \cite{Manuel-Filaci:2020aa}. We had no time to
reach a conclusion.
\bigskip

All this is reviewed in section \ref{sec:mintwist}, together with a
recent result on the
interpretation of the twisted fluctuation of the free Dirac operator as a
torsion term.   In section \ref{sec:sigchange}, inspired by Manuele's
discovery that there exist several  ways to twist the spectral triple
of the Standard Model, we study in full generality the 
inner product induced on $\HH$ by an arbitrary twist (\S
~\ref{subsec:twistinprod}-\ref{subsec:Krein}).  We show in
proposition \ref{prop:Kreinprodcut} that under
rather loose conditions, this product is a Krein product, in agreement
with previous known examples
\cite{Devastato:2020aa,Devastato:2018aa,Martinetti:2019aa} and 
recent studies \cite{Nieuviarts:2024aa}. In
\S\ref{subsec:twistop}  we apply these results to
minimal twists, and point out in \S\ref{subsec:unit} a relation between twisted
unitaries and the symmetry group of twistors.
\newpage 

\section{Minimal twist of the Standard Model}
\label{sec:mintwist}

\subsection{Spectral description of the Standard Model} \label{subsec:sm}

 The spectral triple of the Standard Model \cite{Chamseddine:2007oz}
 is the product 
\begin{align} 
\label{eq:specSM} 
\A= \cinf \otimes
                  \bleu{\A_{F}},\quad \hh= L^2(\M, S)\otimes \hh_{F},
                  \quad D= \ds\otimes \I_F + \gamma_\M\otimes
                  D_{F} \end{align} of the canonical spectral triple
                associated with a  Riemannian, compact, spin manifold
                $\M$ of even dimension $n=2m$, namely
\begin{equation} \label{eq:1} \cinf,\quad
                  L^2(\M, S)\quad \ds=
                  -i\gamma^\mu\nabla_\mu, \end{equation} 
with the finite dimensional spectral  triple 
\begin{align} \label{eq:tripletf} \A_{F}= \C
                         \oplus {\mathbb H} \oplus M_3(\C), \; \;
                         \hh_{F}= \C^{96},\quad D_{F}.
\end{align} 

In \eqref{eq:1}, the algebra of smooth functions on $\M$
acts by multiplication on the Hilbert space $L^2(\M, S)$ of square
integrable smooth sections of the spinor bundle over $\M$. The
operator $\ds$ is the usual Dirac operator of quantum field theory,
where the $\gamma^\mu$'s are the Euclidean Dirac matrices satisfying
the anticommutation relation 
\begin{equation} 
\label{eq:2} 
\left\{\gamma^\mu, \gamma^\nu\right\}= 2g^{\mu\nu}\I 
\end{equation} 
with $g^{\mu\nu}$ the components of the Riemannian metric and $\I$ the
identity operator on $L^2(\M, S)$, while \begin{equation} \label{eq:3}
  \nabla_\mu =\partial_\mu +\omega_\mu \end{equation} is the covariant
derivative on the spinor bundle, with $\omega_\mu$ the spin
connection, that is the lift to the spinor bundle of the Levi-Civita
connection on $\M$.{\footnote{We use Einstein summation on indices in alternate
position: up/down or down/up. In view of the Lorentzian signature, one
counts the dimension of $\M$ by a greek index $\mu$ running from $0$ to $n-1$.}}

The operator $\gamma_\M$ in \eqref{eq:specSM} is the  grading of the
Clifford bundle  on $\M$ (on a
$4$-dimensional manifold, this is the physicist's $\gamma^5$ matrix)
and $\I_F$  the identity endomorphism on $\HH_F$. 

The details of the finite dimensional spectral triple
\eqref{eq:tripletf} are not relevant for this work (they can be found
in \cite{Chamseddine:2007oz}, and are motivated in
\cite{Chamseddine:2008uq}). The operator $D_F$ is a $96\times 96$
matrix containing the parameters of the theory, namely the Yukawa
couplings of fermions, the Dirac and Majorana masses for the
neutrinos, the Cabibbo matrix and the mixing matrix for
neutrinos. What matters in the following  is that $D_F$ splits into a
block diagonal part $D_0$ and an off-diagonal part $D_M$ that commutes
with the representation of the algebra $\A_F$ on $\HH_F$, that
is \begin{equation} \label{eq:5} [\gamma_\M\otimes D_M, a]= 0
  \quad\quad\forall a\in \A=\cinf\otimes \A_{F}. \end{equation} Hence
$D_M$, whose only non zero entry is the Majorana mass $k_M$ of the
neutrino, does not contribute to the $1$-form \eqref{eq:12}.  As explained in the
introduction, we say it is
transparent to fluctuations of the metric. 

To make the neutrino part of the Standard Model contributing to the
generation of bosons, a possibility is to relax the \emph{first-order
  condition}\cite{Chamseddine:2013fk} \begin{equation} \label{eq:4}
  [[D, a], \, JbJ^{-1}]=0\qquad \forall a, b\in\A \end{equation} which
is one of the condition of Connes reconstruction theorem
\cite{connesreconstruct}, also satisfied by the spectral triple of the
Standard Model.  By doing so, one generates the extra-scalar field
$\sigma$ required to stabilize the electroweak vacuum and gets the
correct Higgs mass \cite{Chamseddine:2012fk}. This leads to  a Pati-Salam
model of grand unification \cite{Chamseddine:2013uq}.  

\subsection{Twist by grading} 
\label{subsec:twist-by-grading}

Alternatively, one may consider a bigger algebra $\A'$ acting on the
same Hilbert space $\HH$, but that does 
not commute with $\gamma_\M\otimes D_M$.  In that way, one makes the
Majorana mass $k_R$ of the neutrino contributing to the
fluctuation. It is important to keep the Hilbert space and Dirac
operator untouched, since they encode the fermionic content of the
model and  there is no experimental indication, nor theoretical requirement, of new fermions beyond the Standard Model.

There is a natural candidate for this bigger algebra $\A'$, as soon as
the spectral triple is graded. A grading for a spectral triple $(\A,
\HH, D)$ is a selfadjoint operator $\Gamma$ in $\BH$ which has the properties
of $\gamma_\M$, that is it
squares to $\I$ and is such that 
\begin{equation} 
\label{eq:23} 
\left\{\Gamma, D\right\}=0, \qquad [\Gamma, a]=0 \quad \forall a\in
\A. 
\end{equation} 
Decomposing the Hilbert space as $\HH=\HH_+ \oplus \HH_-$ where $\HH_\pm:= p_\pm \HH$ are
the image of $\HH$ under the projections 
\begin{equation}
  \label{eq:85}
  p_\pm := \frac 12 (\I\pm\Gamma)
\end{equation}
defined by the grading, the representation $\pi$ of $\A$
is block diagonal 
\begin{equation} \label{eq:24} \pi(a)= \begin{pmatrix} \pi_+(a) &0 \\
    0& \pi_-(a) \end{pmatrix} 
\end{equation} 
with $\pi_\pm$ the restriction of $\pi$ to $\HH_\pm$.  Thus $\HH$ is
large enough to carry two independent representations of $\A$, one on
each subspace $\HH_\pm$. In other terms, the Hilbert space of a graded
spectral triple naturally carries a representation $\pi'$ of the
algebra
\begin{equation}
\A'=\A\otimes \C^2,
\label{eq:79}
\end{equation}
namely 
\begin{equation} \label{eq:6} \pi'((a_1, a_2))=\begin{pmatrix} \pi_+(a_1) & 0 \\ 0 & \pi_-(a_2)\end{pmatrix} \qquad \forall a_1, a_2\in \A. \end{equation}

 However, a fundamental property of spectral triples no longer holds
 for $\A'$, namely the boundedness of the commutator $[D, \pi(a)]$ for any
 $a$ in $\A$, This is because $\Gamma$ anticommutes with $D$ by
 definition, so that $D$ is off-diagonal in the eigenbasis of
 $\Gamma$, 
\begin{equation} \label{eq:9} D= \begin{pmatrix} 0 & {\cal D} \\ {\cal
      D}^\dagger & 0\end{pmatrix} 
\end{equation} where $\cal D$ is the restriction of $D$ to $\HH_-$,
with value in $\HH_+$. Therefore 
\begin{equation} \label{eq:10} [D, \pi'(a_1, a_2)] = \begin{pmatrix}
    0& {\cal D} \pi_-(a_2) -\pi_+(a_1) \,{\cal D} \\  {\cal D}^\dagger \,\pi_+(a_1) - \pi_-(a_2)\,{\cal
      D}^\dagger&0 \end{pmatrix}. 
\end{equation} 
This commutator has no reason to be bounded. 

 \begin{exem} For the spectral triple (\ref{eq:specSM}) of the
   Standard Model, the unboundedness comes from the component
   $\ds\otimes\I_F$ of the operator $D$ (the other part
   $\gamma_\M\otimes D_F$ being a bounded operator, its twisted
   commutator with any $a$ is bounded). One has 
\begin{equation} [\ds\otimes\I_F, \pi'((a_1, a_2))]= -i[
  \gamma^\mu\partial_\mu\otimes \I_F , \pi((a_1, a_2))]
  -i[\gamma^\mu\omega_\mu\otimes\I_F , \pi((a_1, a_2))].  \label{eq:19} 
\end{equation} 
The second term is bounded, since the Dirac matrices, the spin
connection and the representation of the algebra are all bounded
operators. Instead, the first  term is not bounded as soon $a_1, a_2$
are non constant. It is enough to check it on 
$a_1=f\otimes {\bf 1}_F$,  $a_2=g\otimes {\bf 1}_F$ 
where $f, g$ are distinct, non constant functions in $\cinf$ and ${\bf 1}_F$ is
the identity of $\A_F$.  Working in dimension $4$ for
simplicity  \linebreak 
\newpage
\noindent(but
the proof extends to arbitrary even dimension) with the Dirac matrices
\begin{equation}
\gamma^\mu =
\begin{pmatrix}
  0 & \sigma^\mu \\ \tilde\sigma^\mu &0 
\end{pmatrix}
\quad\text{ for }\quad  \sigma^\mu =\left\{\I,
  -i\sigma_{j=1,2,3}\right\}, \quad \tilde\sigma^\mu =\left\{\I,
  i\sigma_{j=1,2,3}\right\}\qquad \mu=0,1,2,3
\label{eq:78}
\end{equation}
where the $\sigma_i$'s are the Pauli
matrices, one gets
\begin{align} 
\label{eq:21} 
[-i\gamma^\mu\partial_\mu\otimes\I_F, \pi'(a_1, a_2)] &= -i
                                                        \left(
\begin{pmatrix}
                                                            0 &
                                                            \sigma^\mu \partial_\mu\\
                                                            \tilde\sigma^\mu \partial_\mu&
                                                            0 
\end{pmatrix} 
\begin{pmatrix}
                                                            f & 0 \\0
                                                            &
                                                            g 
\end{pmatrix}
                                                              - 
\begin{pmatrix}
                                                                f& 0
                                                                \\0 &
                                                                g 
\end{pmatrix}
\begin{pmatrix}
                                                                0 &
                                                                \sigma^\mu \partial_\mu\\
                                                                \tilde\sigma^\mu \partial_\mu&
                                                                0 
\end{pmatrix}
                                                                                                \right)\otimes\I_F,\\
                                                      &=-i 
\begin{pmatrix} 0 & \sigma^\mu \partial_\mu g - f
  \sigma^\mu \partial_\mu\\  \tilde\sigma^\mu \partial_\mu f - g
  \tilde\sigma^\mu \partial_\mu&0 
\end{pmatrix}\otimes\I_F.  \end{align}
 By the Leibniz rule, the upper-right block is 
\begin{equation} \label{eq:25} \sigma^\mu (\partial_\mu g) + \sigma^\mu g\partial_\mu - f \sigma^\mu \partial_\mu = \sigma^\mu (\partial_\mu g) + \sigma^\mu (g-f)\partial_\mu . 
\end{equation} 
This is a differential operator, hence unbounded. The same is true for
the lower-left block.\end{exem}

This problem is similar as the one treated in
\cite{Connes:1938fk}, yet  in a different context (see remark
\ref{rem:conform}). The solution is to substitute the
boundedness of the commutator of $D$ with $\A$, with the requirement
that there exists an automorphism $\rho$ of $\A'$ such that the
\emph{twisted commutator} \begin{equation} \label{eq:11} [D,
  a']_\rho:=Da' - \rho(a')D \end{equation} is bounded for any
$a'\in\A'$. This defines a \emph{twisted spectral triple}
\begin{equation}
(\A', \HH,
D), \rho
\label{eq:70}
\end{equation}
(also called $\sigma$-spectral triple, but we dot use this
terminology to avoid confusion with the extra-scalar field $\sigma$
obtained by fluctuating of the Majorana mass of the neutrino). 

In our case, $\rho$ is the automorphism of $\A\otimes \C^2$
that flips the two elements of $\A$, 
\begin{equation} 
\label{eq:13} 
\rho((a_1, a_2)):= (a_2, a_1)\quad \forall (a_1, a_2)\in \A\otimes
\C^2. 
\end{equation} 
Indeed, (restoring the symbol of representation) the boundedness of  
\begin{equation*} 
[D, \pi'(a_1, a_2)]_\rho = D\,\pi'(a_1, a_2) -
  \pi'(a_2, a_1)\, D= \begin{pmatrix} 0& {\cal D} \pi_-(a_2) - \pi_+(a_2) {\cal D} \\
    {\cal D}^\dagger \pi_+(a_1) - \pi_-(a_1){\cal
      D}^\dagger&0 \end{pmatrix} \end{equation*} 
follows from the
boundedness of \begin{equation} [D, \pi(a_1)]= [D, \pi'(a_1,
  a_1)]=\begin{pmatrix} 0 &{\cal D} \pi_-(a_1) - \pi_+(a_1) {\cal D} \\  {\cal
      D}^\dagger \pi_+(a_1) - \pi_-(a_1) {\cal D}^\dagger&0\end{pmatrix} 
\end{equation}
 (and similarly for $a_2$) inherited  from the definition of the initial spectral triple $(\A, \HH, D)$.

If the representations $\pi_\pm$ in  \eqref{eq:24} are faithful, the
same is true for \eqref{eq:79} and we obtain a  twisted spectral
   triple \begin{equation} \label{eq:14} (\A\otimes\C^2, \, \HH, \, D),
     \,\rho\end{equation} called a \emph{minimal twist} of the graded
   spectral triple $(\A, \HH, D)$.  
 The adjective ``minimal'' reflects the fact that the Hilbert space and the operator $D$ are untouched, only the algebra is modified. 

\begin{rem} 
\label{rem:conform} In \cite{Connes:1938fk}, a twisted spectral triple is
  associated to the canonical spectral triple \eqref{eq:1} of a
  manifold by applying
  a conformal map on the Dirac operator (keeping the same algebra and
  Hilbert space). It satisfies the so called Lipschitz
  condition (omitting again the symbol of representation) 
\begin{equation} \label{eq:26} [|D|, a']_\rho \:\text{ is bounded } \; \forall a'\in \A'.  
\end{equation} The minimal
  twist \eqref{eq:14} does not satisfies this condition, making it an interesting
  candidate to apply the double index formula introduced in \cite[\S
  3]{Connes:1938fk}. This will be explored in a forecoming paper
  \cite{Martinetti:2025aa}. 
\end{rem} 

\newpage

 Any graded spectral triple $(\A, \HH, D)$ such that the
representations $\pi_\pm$ in \eqref{eq:24} are faithful can be minimally twisted
following the procedure above, that we call a \emph{twist by
  grading}.  The general properties of the resulting twisted spectral
triple  have been studied in \cite{Lett.}. In particular, it is graded
with the same grading as the initial triple. Furthermore, if the
initial spectral triple is real, then its twisted partner is real with
the same real structure and $KO$-dimension.  If the initial spectral triple satisfies the
first order condition, then its twisted partner automatically
satisfies a twisted version of it, 
\begin{equation} [[D, a']_\rho, \, Jb'J^{-1}]_{\rho_0}=0\qquad \forall
  a', b'\in\A\otimes\C^2\label{eq:18} \end{equation}
 where 
\begin{equation} \label{eq:22} \rho_0(JbJ^{-1}):= J \rho(b) J^{-1}.  
\end{equation} 
The fluctuations of the metric as well make sense in the twisted
context, and amounts to substituting the generalised $1$-form $A$ in
\eqref{eq:7} with a generalised twisted $1$-form $A_\rho$, that is an
element of 
\begin{equation} \label{eq:15} \Omega^1_{D}(\A,\rho):= \left\{ \sum_i
    a_i \, [D, b_i]_\rho, a_i, b_i\in
    \A\otimes\C^2\right\}.  \end{equation} 
This yields a \emph{twisted-covariant} Dirac operator 
\begin{equation} \label{eq:17} D_{A_{\rho}}= D +A_\rho + J A_\rho J^{-1}.  \end{equation}

\subsection{Minimally twisted Standard Model}
\label{subsec:minsm} 

By twisting the spectral triple \eqref{eq:specSM} of the Standard
Model, the hope was to be able to generate the extra-scalar field
$\sigma$ discussed at the end of \S \ref{subsec:sm} thanks to a
twisted-fluctuation of the metric~\eqref{eq:17}, while preserving the first order
condition in its twisted form \eqref{eq:18}.  However, and Manuele was
the first to notice it, the twist using the grading of the spectral triple  (\ref{eq:specSM}),
\begin{equation}
\Gamma=\gamma_\M\otimes \Gamma_F,
\label{eq:71}
\end{equation}
(with  $\Gamma_F$ a
grading  of the spectral triple (\ref{eq:tripletf})) does
not permit this, for $\gamma_\M\otimes D_M$ twist-commutes with
$\A\otimes \C^2$: \begin{equation} \label{eq:28} [\gamma_\M\otimes D_M , \pi'(a_1, a_2)]_\rho=0.  \end{equation}  This is because the diagonal part $D_0$ and
off-diagonal parts $D_M$ of $D_F$ anticommutes independently with
$\Gamma_F$ (see \cite[\S 4.2]{Filaci:2023aa}). In other
terms, $\gamma_\M\otimes D_M$ is also transparent under twisted fluctuations.

This argument holds for the twist by grading, but fortunately, the 
twist does not necessarily need to be by grading. For any \emph{almost commutative
  geometry}, that is any product \eqref{eq:specSM}  of the
canonical spectral triple \eqref{eq:1}  of $\M$ with a finite dimensional
spectral triple (non necessarily \eqref{eq:tripletf}, one may repeat the
construction of \S \ref{subsec:twist-by-grading} using instead of the
grading $\Gamma$ any operator 
\begin{equation} 
T= \gamma_\M\otimes T_F\label{eq:16} 
\end{equation} 
with $T_F$  a bounded operator on $\HH_F$  distinct from 
$\pm\I_F$ and  such that 
\begin{equation} T_F=T_F^\dagger,\qquad T_F^2=\I_F \quad \text{ and
  }\quad [T_F, m]=0 \;\,\forall  m\in \A_F.
\label{eq:27} 
\end{equation} 
These conditions guarantee that $\HH$ carries two independent
representations of $\A$, one on each of the eigenspaces of $T$.
This allows to build the minimal twist with twisting automorphism the
flip \cite{Manuel-Filaci:2020aa}.  

\newpage
\noindent But there is no need that 
\begin{equation} 
\left\{T_F, D_F\right\}=0, 
\label{eq:35} 
\end{equation} 
meaning that $\Gamma_F$ does not need to be a grading of the spectral triple \eqref{eq:tripletf}. The advantage of using $T_F$ instead of $\Gamma_F$ is that as soon as \eqref{eq:35} does not hold, then \eqref{eq:28} has no reason to hold neither, hence $\gamma_\M\otimes D_M$ is no longer transparent under twisted fluctuations.

A classification of all possible twisting operators $T_F$ for 
arbitrary almost commutative geometries 
 has been undertaken in \cite{Manuel-Filaci:2020aa} (see also \cite[\S 4.2]{Lett.}). It is not yet fully achieved, and we still do not know in how many non-equivalent ways the Standard Model can be minimally twisted. 
Nevertheless, we know that there exists at least one twisting operator
which permits to obtain the extra-scalar field $\sigma$, namely the
operator $T_F$ which takes value \begin{equation} +1\; \text{on left
    particles and antiparticles},\quad -1 \;\text{on right particles
    and antiparticles.}\label{eq:29} \end{equation} 
 Unfortunately, this minimal twist  violates the twisted first-order
 condition \eqref{eq:18}. However, this does not prevent to fluctuate
 the metric, following the procedure worked out in
 \cite{Martinetti:2021aa}. The latter generalises to twisted spectral
 triples the ``fluctuation without first order condition'' introduced
 in  \cite{Chamseddine:2013fk}.

\begin{rem} Attributing the positive (negative) sign to left (right)
  handed entities is a matter of convention. What matters here is that
  $T_F$ is not the grading $\Gamma_F$ because the latter takes the
  same value on right particles \& left antiparticles, and on left
  particles \& right antiparticles.  \end{rem}
 The field content of this model has been developed in~\cite{Filaci:2021ab}: \begin{itemize} \item The twisted fluctuation of the diagonal part of $\gamma_\M\otimes D_F$ yields two Higgs doublets which act independently on the chiral components of Dirac spinors.  They are expected to mix in a single linear combination in the action, yielding the usual Higgs doublet.  \item The off-diagonal part of $\gamma_\M\otimes D_F$ yields a doublet of chiral extra-scalar fields.  \item The twisted fluctuation of $\ds\otimes \I_F$ yields the bosonic part of the Standard Model. In addition, one also gets three new $1$-form fields: two with value in $\R$, one with value in $M_3(\C)$.  \end{itemize}

\noindent Some preliminary results on the action were worked out by Manuele in its PhD, but not published.

\bigskip It is not yet clear whether there exists a twisting operator
that allows to obtain the extra-scalar field while preserving the
first-order condition \eqref{eq:18}, or if the generating of an
extra-scalar field and the twisted first-order condition are mutually
exclusive. In the latter case, one may legitimately question the interest of all that for physics. Since relaxing the first-order condition for usual (i.e. non-twisted) spectral triples is enough to get the desired extra-scalar field (as shown in \cite{Chamseddine:2013uq}), then
\begin{center}
  why twist it ?
\end{center}
\smallskip

Actually the truest interest of twist relies in the new fields of
$1$-forms obtained from the twisted fluctuations of $\ds\otimes\I_F$.
Similar forms already appeared in \cite{buckley}, where only the weak
interaction part of the Standard Model was twisted. This model had some mathematical issue
discussed in \cite{Filaci:2023aa}, later solved in
\cite{Filaci:2021ab}. The meaning of these $1$-forms was unclear at
the time, and has been only recently elucidated in two different
contexts: as torsion in the minimal twist of a manifold, as a change
of signature in the minimal twist of electrodynamics. We briefly
review the torsion result in the next subsection. The change of
signature is discussed at length in section \ref{sec:sigchange}.





\subsection{Torsion as a twisted fluctuation} 
The new fields of $1$-form discussed above have  their origin in the twisted
fluctuation  of
$\ds\,\otimes\,\I_F$.  In order to understand their geometrical meaning, it seems reasonable to focus on the manifold part of the spectral triple.

The twist by grading of the canonical spectral triple \eqref{eq:1} of
a Riemannian, compact, spin manifold $\M$ of even dimension $2m$ is 
\begin{equation} 
\label{eq:184} 
\A= \cinf\otimes\,\C^2, \quad \HH= L^2(\M,S), \quad D=\ds; \quad \rho 
\end{equation} 
with representation
\begin{equation} 
\pi(f, g)=\left(\begin{array}{cc} f\,\I_{2^{m-1}}& 0 \\  0&
                                                           g\I_{2^{m-1}}\end{array}\right)
\label{eq:187} 
\end{equation}
and  twisting automorphism 
\begin{equation}
  \label{eq:42}
  \rho(f,g) = (g,f) \quad\quad \forall (f,g)\in \A. 
\end{equation}
This is the only way to minimally twist a manifold by a finite
dimensional algebra \cite[Prop. 4.2]{Lett.}.

In $KO$-dimension $0$ and $4$ (that is for a manifold of dimension
$2m= 0 \text{ or } 4 \text{ modulo } 8$), there exist non-zero
selfadjoint twisted fluctuations of the Dirac operator~$\ds$, yielding
the twisted covariant operator \cite{Lett.} 
\begin{equation} D_{\omega_f}= \ds -i\, f_\mu\gamma^\mu\gamma_\M\quad
  \text{ with } \quad f_\mu \in C^{\infty}(\M, \R)\label{eq:30}. 
\end{equation} 
The index notation $\omega_f$ is justified noting that the second
term in the equation above is the Clifford action $c$ of the Hodge
dual of the $1$-form 
\begin{equation} 
\label{eq:32} 
\omega_f= f_\mu dx^\mu, 
\end{equation} 
namely one has \cite{Martinetti:2025ab} 
\begin{equation} -i\, f_\mu\gamma^\mu\gamma_\M=\frac{(-i)^{m+1}}{2m}\,
  c(\star \omega_f).
\label{eq:31} 
\end{equation} 
This is a remarkable aspect  of the minimal twist, for in the non
twisted case all the fluctuation of $\ds$ vanish \cite{Connes:1996fu}.  

\begin{rem} 
In $KO$-dimension $2$ and $6$, all selfadjoint twisted fluctuations of
$\ds$ vanish. In odd dimension, there is no grading and the whole
construction would need to be suitably adapted.
\end{rem} 

Furthermore, in dimension $4$, one shows that $D_{\omega_f}$ in
(\ref{eq:30}) is the lift to spinors of an orthogonal and geodesic
preserving connection, with torsion $3$-form
\begin{equation}
{\cal T}=-\star\omega_f.
\label{eq:43}
\end{equation}
In
other dimensions $4 \, [\text{mod }4]$, it is not clear if  there
exists a similar geometrical interpretation of the extra term, other than
\eqref{eq:31}. For instance on a manifold of dimension  $8$, ${\cal T}$ is a
$7$-form to which one may associate tensors of type $(p, q)$ with
$p+q=7$. Has any of  them a particular geometrical meaning ? 

In the minimal twist of the Standard Model, one expects that the new
$1$-form fields have a similar interpretation: two of them as torsion,
the third one as a $M_3(\C)$-valued torsion.

\newpage
\section{Signature change}
\label{sec:sigchange}

\subsection{Twisted inner product}
\label{subsec:twistinprod}

The link between  minimal twists and a transition from the Euclidean
to the Lorentzian has been first stressed in
\cite{Devastato:2018aa}. It was noticed that if  the twisting
automorphism  of a twisted spectral triple $(\A, \HH,
D), \rho$  (with representation $\pi$) extends to an inner automorphism of
the algebra $\B(\HH)$ of bounded operators on $\HH$, that is if there exists a unitary $R\in \B(\HH)$ such that
\begin{equation}
  \label{eq:48bis}
  \pi(\rho(a)))= R\, \pi(a)\, R^* \quad \forall a\in\A,
\end{equation}
then $\rho$ induces an inner product on $\HH$, called \emph{twisted product},
\begin{equation}
  \label{eq:48}
  (\psi, \varphi)_R := <\psi, R\varphi>
\end{equation}
where $\langle\cdot, \cdot\rangle$ is the initial Hilbert product on
$\HH$. The point is that the twisted product is non necessarily
positive definite.

For instance, for the minimal twist of a manifold \eqref{eq:184}, a natural choice
for $R$ is the first Dirac matrix
\begin{equation}
  \label{eq:38}
  \gamma^0=
  \begin{pmatrix}
    0& \I_{2^{m-1}}\\\I_{2^{m-1}}& 0
  \end{pmatrix}.
\end{equation}
 The twisted product is then (with  $dv$ the volume form on $\M$)
\begin{equation}
\label{eq:39bis}
  (\psi, \varphi)_R =\int_\M \langle\bar\psi(x) ,\,\gamma^0 \,\varphi(x) \rangle_{p}\, dv
\end{equation}
where $\langle \cdot, \cdot\rangle_p$ is the canonical scalar product
on $\C^p$ (for $p=2^m$) and $\bar\psi$ is the complex conjugate of
$\psi$. The integrand  is  the one of the inner product between spinors on a
\emph{Lorentzian manifold} $\tilde\M$, namely 
\begin{equation}
  \label{eq:39}
\int_{\tilde\M} \langle\bar\psi(x) ,\,\gamma^0 \,\varphi(x) \rangle_{p}\, d\tilde v.
\end{equation}
Similar considerations for arbitrary
twisted spectral triples satisfying~\eqref{eq:48bis} are found in \cite{Nieuviarts:2024aa}.

Beware that the manifold $\M$ is not turned into a
Lorentzian manifold $\tilde\M$, so the integral  (\ref{eq:39bis}) does not concides
with \eqref{eq:39}. Nevertheless, if one considers a local region of integration, then by Wick rotation the Riemannian
and Lorentzian volume forms should be equal up to a sign. As well 
\eqref{eq:39bis} and \eqref{eq:39} are equal up to a sign, and yield
the same equations of motion.
 If one
integrates on the whole manifold however, one is confronted to the
question of turning a Riemannian manifold into a Lorentzian one,
which is known not to be always possible.  So for the moment, one can
only say that the twist induces \emph{a change of signature  in  the
integrand of the inner product}. Even though limited, this change has
significant consequences in the fermionic action \cite{Martinetti:2019aa,Martinetti:2025ab}.

\begin{rem}
To avoid confusion, let us stress that during the twisting
procedure, the Dirac operator is always a Riemannian one.  We do not
deal with a Lorentzian Dirac operator, contrary to other approches to Lorentzian noncommutative geometry, like
  \cite{Strohmaier:2006xi,Barrett:2007vf,Franco:2012fk,Dungen:2015aa,Nieuviarts:2025aa}.
\end{rem}

In next section,  we generalize the condition \eqref{eq:48bis}  and show that any twisted
automorphism that extends to the whole of ${\cal B}(\HH)$ necessarily
induces a twisted product, but non necessarily implemented by a
unitary.  
 This allows to deal with  twisting
automorphisms $\rho$ that are non necessarily $*$-automorphisms, in agreement
with the unitarity condition  \cite{Connes:1938fk}
\begin{equation}
  \label{eq:63}
  \rho(a^*)= \left(\rho^{-1}(a)\right)^*.
\end{equation}
In \S \ref{subsec:Krein} we work out some rather loose conditions
guaranteeing that a twisted product induces a Krein structure. We
apply these results to minimal twists  (non necessarily by grading)
 in \S \ref{subsec:twistop}  and conclude with some considerations 
on twisted
 unitaries in \S \ref{subsec:unit} . 


\subsection{Expandable twist}
\label{subsec:exp}

Let us begin with a suitable generalisation of the
condition  (\ref{eq:48bis}).
\begin{defi}
 \label{def:expandtwist}
 The twisting automorphism $\rho$ of a twisted spectral triple $(\A,
  \HH, D), \rho$ is \emph{expandable} if there exists an
  automorphism $\tilde\rho$ in $\text{Aut}(\BH)$ such  that
  \begin{equation}
    \label{eq:77}
    \pi(\rho(a)) = \tilde\rho(\pi(a))\quad \forall a\in\A.
  \end{equation}
By extension, we say that the twisted spectral triple itself is \emph{expandable}.
\end{defi}
Given an algebra $\A$ acting as
bounded operators on an Hilbert space $\HH$ through a representation $\pi$,
an automorphism of $\A$
is \emph{spatial}  \cite{Okayasu:1968aa} if there exists an invertible operator $R$ in $\BH$ such that 
\begin{equation}
  \label{eq:72}
  \pi(\rho(a))= R \, \pi(a)\, R^{-1} \qquad \forall a\in \A.
\end{equation}
 \begin{prop}
\label{prop:prodherm}
Any expandable twisting automorphism is spatial, and thus induces a
non-degenerate 
twisted product. The latter is  Hermitian iff
 $\rho$ is spatial for some selfadjoint $R$ in $\BH$.
\end{prop}
\begin{proof}
By Okayasu \emph{polar decomposition theorem} for von Neumann algebra automorphisms 
\cite[Th.9]{Okayasu:1968aa}, any
automorphism $\tilde \rho$ of a von
Neumann algebra $M$, acting on an Hilbert space $\HH$, is the composition
\begin{equation}
\tilde\rho =\tilde\rho_2\circ\tilde\rho_1\label{eq:80}
\end{equation}
of an  automorphism 
$\tilde\rho_1(\cdot)= P\cdot P^{-1}$ induced by an invertible,
positive operator $P\in M$ {\footnote{Called inner by Okayasu.}}
with a $*$-automorphism $\tilde\rho_2$. For $\M=\BH$, any
$*$-automorphism is inner (e.g. \cite[II.5.5.14]{blackadar2006}):
$\tilde\rho_2(\cdot)= U\cdot U^\dag$ for some unitary $U\in\BH$.
Hence any automorphism $\tilde\rho$ of
$\BH$ is of the form 
\begin{equation}
  \label{eq:81}
  \tilde\rho(\cdot) = R \cdot R^{-1}
\end{equation}
for some $R=UP$ in $\BH$.

  The induced twisted product  \eqref{eq:48} is non-degenerate, for 
\begin{equation*}
(\psi, \varphi)_R = 0\:\forall \varphi\in\HH \Longleftrightarrow \langle \psi,
R\varphi  \rangle =0 \:\forall \varphi \Longleftrightarrow \langle \psi,
\tilde\varphi  \rangle =0 \:\forall \tilde\varphi \in\text{Im}(R)\Longleftrightarrow \langle \psi,
\tilde\varphi  \rangle =0 \:\forall \tilde\varphi \in\HH
\end{equation*}
where we use that $R$
 is invertible, so has  image $\HH$.
It is hermitian 
if, and only if, 
\begin{align}
  \label{eq:58}
0=\overline{(\psi, \varphi)_\rho} -(\varphi, \psi)_\rho  &=
\overline{\langle \psi, R\varphi\rangle}-\langle \varphi,
  R\psi\rangle,\\
&={\langle  R\varphi,\psi \rangle}-\langle \varphi, R\psi\rangle
=\langle\varphi, (R^\dag -R), \psi\rangle \quad \forall \varphi, \psi\in\HH,
\end{align}
that is $R=R^\dag$. 
\end{proof}
\begin{cor}
  An expandable  twisting automorphism that induces an Hermitian twisted product is
  necessarily unitary in the sense of (\ref{eq:63}).
\end{cor}
\begin{proof}
By the previous proposition, an expandable  twisting automorphism
$\rho$ satisfies \eqref{eq:72}.
Since the representation
$\pi$ is involutive ($\pi(a^*)=\pi(a)^\dag$), one has 
\begin{align}
  \pi(\rho(a^*)) =R\,\pi(a)^\dag\, R^{-1}
\label{eq:unitcondd}
\end{align}
while, using  $\pi(\rho^{-1}(a))=R^{-1}\pi(a) R$
(obtained substituting $a$ with $\rho^{-1}(a)$ in \eqref{eq:72})
\begin{align}
\pi((\rho^{-1}(a))^*)=\pi(\rho^{-1}(a))^\dag= \left( R^{-1} \,\pi(a)
 \, R\right)^\dag= R^\dag\, \pi(a)^\dag\, (R^{-1})^\dag. 
\label{eq:unitcondd2}\end{align}
If the twisted product is Hermitian, then by the previous proposition
$R$ is selfadjoint, thus \eqref{eq:unitcondd} is equal to \eqref{eq:unitcondd2}.
Since $\pi$ is faithful, this implies \eqref{eq:63}. \end{proof}

The converse is not necessarily true. If $\rho$ satisfies
\eqref{eq:63}, then
from  \eqref{eq:unitcondd}  and  \eqref{eq:unitcondd2}  one gets that
$R^{-1}R^\dag$ belongs to the commutant of $\pi(\A)$, but is not
necessarily equal to $\I$.
\newpage

\subsection{Krein structure}
\label{subsec:Krein}

A  Krein space is a vector space $\cal K$ equipped with an
undefinite, non-degenerate, inner product $(\cdot, \cdot)$, that
decomposes as a direct sum 
\begin{equation}
\K =\K_+ \oplus
\K_-\label{eq:75}
\end{equation}
such that  the product is
positive, resp. negative, definite on $\K_+$, resp. $\K_-$,  and the
latter are complete  with respect to the norm induced
by $\pm (\cdot, \cdot)$.  A \emph{fundamental symmetry} is any
operator $F$ on $K$ with
\begin{equation}
F^2=\I \quad \text{ and  } \quad (\cdot , F\cdot)=(F\cdot , \cdot)
\label{eq:108}
\end{equation}
such  that the
product $(\cdot , F \cdot)$ is definite  positive. For instance 
\begin{equation}
F=\I_+ \oplus(-\I_-)
\label{eq:95b}
\end{equation}
where $\I_\pm$ the identity
operators  on $\K_\pm$. The space $\cal K$ equipped with the product
$(\cdot, F\cdot)$ is then a Hilbert space.
 
The seminal example of Krein space is the space of square integrable
sections of the spinor bundle on a Lorentzian manifold, with product
\eqref{eq:39}.
As explained in \S \ref{subsec:twistinprod}, this Krein product coincides
with the twisted product of the minimally twisted manifold. 

\begin{rem} 
\label{rem:stroh}
Following \cite[eq. 21] {Strohmaier:2006xi}, on a pseudo-Riemannian
  manifold of dimension $n$ with signature $(n-k, k)$ where $k$ is the
  number of $-1$, then the Krein structure on the fibers
  of the spinor
  bundle is
  \begin{equation}
\langle
  \cdot\,, \frak J\,\cdot\rangle_p\label{eq:44}
  \end{equation}
where $\langle \cdot,
  \cdot\rangle_p$ is the canonical scalar product on $\C^p$ ($p$ the
dimension of the spin representation) and 
\begin{equation}
  \label{eq:99}
  \frak J = i^{\frac{n(n-1)}{2}}\gamma^1... \gamma^k.
\end{equation}
In particular, on a Lorentzian manifold of signature $(1, -1, -1,-1)$
one gets 
(counting the dimension of $\M$ from $0$ to $3$)
\begin{equation}
  \label{eq:100}
   \frak J = i^{\frac{k(k+1)}{2}}\gamma_L^1 \gamma_L^2 \gamma_L^3= i
   \begin{pmatrix}
     0& \I_2\\ -\I_2 &0
   \end{pmatrix}
\end{equation}
where $\gamma_L^j=i\gamma^j$ are the Lorentzian Dirac matrices
obtained multiplying by $i$ the Euclidean ones \eqref{eq:78}, and we
recall that the Pauli matrices are
\begin{equation}
  \label{eq:101}
  \sigma_1=
  \begin{pmatrix}
    0&1\\ 1&0
  \end{pmatrix},\qquad   \sigma_2=
  \begin{pmatrix}
    0&-i\\ i&0
  \end{pmatrix},\qquad   \sigma_3=
  \begin{pmatrix}
    1&0\\ 0&-1
  \end{pmatrix}.
\end{equation}
The Krein product $\langle
  \cdot, \frak J\cdot\rangle_p$  is equivalent with the twisted
  product \eqref{eq:39bis} of the minimally twisted manifold, for
  $\frak J$ is unitary equivalent to $\gamma^0$ for
  \begin{equation}
U=\frac 1{\sqrt 2}
\begin{pmatrix}
  \I_2& \I_2\\ -i\I_2& i\I_2
\end{pmatrix}.\label{eq:36}
\end{equation}
\end{rem}
\newpage

The main result of this section is a sufficient condition for a
twisted product to be Krein. 
\begin{prop}
\label{prop:Kreinprodcut}
 Let $(\A, \HH, D), \rho$ be an expandable twisted spectral
 triple in the sense of definition \ref{def:expandtwist}. Among the Hermitian twisted products, all those induced
 by an operator $R$ with pure point spectrum, and such that
 $0$ is not a limit point of the spectrum, are Krein products.
\end{prop}
\begin{proof}
The operator $R$ is invertible and selfadjoint by
proposition \ref{prop:prodherm},  and has pure point spectrum by hypothesis, so
its eigenvalues are real, non zero and the corresponding
eigenvectors span $\HH$ (see e.g. \cite[Th. VII.4]{Reed1980}). The latter is thus the direct sum
\begin{equation}
  \label{eq:97}
  \HH= \HH^+ \oplus \HH^-
\end{equation}
of the eigenspaces of $\HH_\pm$ corresponding to positive\slash negative 
eigenvalues. Any $\psi^+\in \HH^+$ is the sum 
\begin{equation}
  \label{eq:102}
  \psi^+=\sum_{i=1}^{n_I}\sum_{k=1}^{n_i} \alpha_{ik}
  \,\varphi_{ik}\qquad \alpha_{ik}\in\R
\end{equation}
where $n_I$ is the number of positive eigenvalues
of $R$, and for any distinct eigenvalues $\lambda_i$ with degeneracy
$n_i$ the set $\left\{\varphi_{ik}\in\HH,\,  k=1, ..., n_i\right\}$
is a orthonormal basis of the corresponding  eigenspace. Note that
$n_I$ and the $n_i$'s may well be infinite.  
 For $\psi^+\neq 0$, one checks that 
\begin{equation}
  \label{eq:106}
 (\psi^+ ,\psi^+)_R= \langle\psi^+, R\psi^+\rangle =\sum_{i,k} |\alpha_{ik}|^2 \lambda_i
\end{equation}
is strictly positive, meaning that the (non-degenerate) twisted product restricted  to
$\HH_+$ is definite positive, and induces a norm
$||\cdot||^+=\sqrt{(\cdot, \cdot)_R}$. From \eqref{eq:106} one has
 \begin{equation}
   \label{eq:107}
   (||\psi^+||^+)^2 =\langle\psi^+, R\psi^+\rangle^2\geq \lambda_m
   \sum_{i,k} |\alpha_{ik}|^2 =\lambda_m ||\psi^+||^2
 \end{equation}
where $\lambda_m\neq 0$ is the smallest positive eigenvalues and the
norm on the r.h.s. is the one on $\HH$. Therefore any sequence
$\left\{\psi^+_n\right\}$ in $\HH^+$ which is Cauchy for $||\cdot
||^+$ is also Cauchy for $||\cdot ||$. Furthermore, denoting
$R^+:\HH^+\to\HH^+$ the
restriction of $R$ to $\HH^+$, one has 
\begin{equation}
  \label{eq:103}
(|| \psi^+||^+)^2=\langle \psi^+\oplus 0, (R^+\oplus\I_-)\psi^+\oplus 0\rangle \leq \vert\vert R^+
  \vert\vert^2 \, ||\psi^+||^2.
\end{equation}
Thus any Cauchy sequence $\left\{\psi^+_n\right\}$ in $\HH^+$ complete for $||\cdot||$ is also complete
for $||\cdot||^+$. This shows that $||\cdot ||^+$ is complete on
$\HH^+$. 

Similarly, the twisted product restricted to
$\HH^-$ is definite negative, and $\HH^-$ is complete for the
norm $||\cdot||^-=\sqrt{-(\cdot, \cdot)_R}$. Thus $(\cdot, \cdot)_R$ is a Krein product.
\end{proof}

\noindent  The initial product on $\HH$ is retrieved as 
\begin{equation}
  \label{eq:110}
  \langle\,\cdot, \cdot\rangle= (\cdot,R^{-1} \cdot)_R
\end{equation}
but $R^{-1}$ is not a fundamental symmetry, unless $R$
is unitary. Then $R=R^*= R^{-1}$ and the conditions \eqref{eq:108} are
satisfied. 

\medskip

Proposition \ref{prop:Kreinprodcut} gives a criteria that guarantees the
existence of a twisted Krein product, it does not says that such a
product necessarily exists for any  expandable twisted spectral triple. In the next section, we show that it
does exist in all the examples of twisted spectral triples relevant in high energy physics.
\newpage

\subsection{Minimal twisting by an operator}
\label{subsec:twistop} 

As explained   in \S \ref{subsec:twist-by-grading}, it is always possible to associate a twisted partner to a graded spectral
triple whose representation is sufficiently faithful. But, as
illustrated in the example of the Standard Model, there may well be
other way to do it.  The following definition fixes the frame.
\begin{defi}
\label{def:almostwist}
Let $(\A, \HH, D)$ be a spectral triple and $T$ a selfadjoint,
 bounded, involutive ($T^2=\I$)
 operator on $\HH$, not a multiple of the identity, that commutes with $\pi(\A)$. Let
 \begin{equation}
   \label{eq:62}
   p_\pm = \frac 12(\I \pm T)
 \end{equation}
be the projection on the two orthogonal eigenspaces $\HH_\pm$ of $T$
{\footnote{$T$ being unitary and selfadjoint,  has a discrete spectrum
    $\left\{-1, 1\right\}$. The elements of the latter are thus
    eigenvalues and $\HH$ is the direct sum of the corresponding
    eigenspaces $\HH_+$, $\HH_-$ (e.g. \cite[Prop. 4.4.5]{Pedersen:1989fk}).}} and  \begin{equation}
      \label{eq:34}
   \pi'((a_1, a_2)):= p_+  \pi(a_1) +
p_-\,\pi(a_2) \qquad \forall a_1, a_2\in \A
    \end{equation}
the representation \eqref{eq:24}  of
 $\A\otimes\C^2$ on $\HH$. If the 
latter is faithful and the twisted commutator
\begin{equation}
[D, \pi'(a_1, a_2)]_\rho \quad \text{ for }  \rho\text{ the 
 flip }  \eqref{eq:13}
\label{eq:113}
\end{equation}
is bounded for any $(a_1, a_2)\in\A\otimes\C^2$,
 then the twisted spectral
 triple
 \begin{equation}
(\A\otimes \C^2, \HH, D),\,\rho
\label{eq:83}
\end{equation}
is called \emph{the minimal
  twist of $(\A, \HH, D)$ by $T$.}
\end{defi}
 
\noindent $T$ is asked to commute with the algebra in order to guarantee
that $p_+\pi$ commutes with $p_-\pi$, so that $\pi'$ is a representation  (otherwise $\pi'((a_1, a_2))\, \pi'((b_1, b_2))$ may not be equal to $\pi'((a_1b_1, a_2b_2))$ because
of non vanishing cross terms).

To associate a twisted product \eqref{eq:48} with  a minimal
twist, the latter must be expandable.
\begin{prop} 
\label{prop:Krein} 
The minimal twist
  of  any spectral triple $(\A, \HH, D)$ by some operator $T$ is
  expandable  in the sense of definition \ref{def:expandtwist} if, and only if, there exists an operator $R$ such that 
  \begin{equation}
    \label{eq:112}
    \left\{R, T\right\}=0, \qquad [R, \pi(a)]=0\quad\forall a\in\A.
  \end{equation}
\end{prop}

\begin{proof}
Writing
\eqref{eq:34} for  $a_1=-a_2={\bf 1}$ the unit of~$\A$, one gets
\begin{equation}
  \label{eq:88}
T=  \pi'(({\bf 1}, {\bf -1})).
\end{equation} 
  Assume that $\rho$ is expandable, that is there exists  $\tilde\rho$
  as in \eqref{eq:81} that satisfies \eqref{eq:77} (written for $\pi'$).
\linebreak   In particular 
\begin{equation}
      \label{eq:40}
     RTR^{-1}=  R\, \pi'(({\bf 1}, {\bf -1}))\, R^{-1} =\pi'((-{\bf 1},
{\bf 1})) =-T. \end{equation}
Hence the first  equation \eqref{eq:112}.
The latter implies $Rp_\pm= p_\mp R$, so that by \eqref{eq:77}
\begin{align}
  \label{eq:114}
0&= \pi'((a, a))- R\,\pi'((a, a))\, R^{-1},\\
 \label{eq:114bis}
& = p_+\left(\pi(a)
-R\pi(a)R^{-1}\right) + p_-\left(\pi(a) - R\pi(a)R^{-1}\right)=
  \pi(a) - R\pi(a)R^{-1} \quad \forall a\in\A.
\end{align}
Hence the second equation \eqref{eq:112}.

Conversely, any $R$ satisfying \eqref{eq:112} extends the twist, for 
\begin{align}
  \label{eq:115}
  R\pi'((a_1, a_2))R^{-1} &= R\, p_+\pi(a_1) \,R^{-1} +R\,p_-\pi(a_2)
                            \, R^{-1},\\
\nonumber =
 & p_-\,R\pi(a_1) R^{-1} +p_+\,R\pi(a_2) R^{-1}=  p_-\,\pi(a_1) 
  +p_+\,\pi(a_2) =\pi'((a_2, a_1)).
\end{align}

\vspace{-.7truecm}\end{proof}

A  minimal twist by $T$ is expandable only if the eigenspaces $\HH_\pm$ of
$T$ have
the same dimension. Indeed, in the decomposition  $\HH= \HH_+\oplus\HH_-$, the operators $T$ and $R$
 are block matrices
    \begin{equation}
      \label{eq:37}
 T=
 \begin{pmatrix}
   \I_+ & 0 \\ 0& -\I_-
 \end{pmatrix},\quad
\quad R=
      \begin{pmatrix}
        R_{++} & R_{+-}\\ R_{-+} & R_{--}
      \end{pmatrix}
    \end{equation}
where $\I_\pm$ are the identity operators on $\HH_\pm$ and  $R_{ij}$
is an operator from $\HH_i$ to $\hh_j$ for $i,j\in\left\{ +,
  -\right\}$. The anticommutation of $T$ and $R$ implies   
  $R_{++}= R_{--}=0$. By definition $R$ is invertible on the left and
  right (otherwise \eqref{eq:72}  or its inverse 
$R^{-1}\cdot R$  would not be algebra morphisms). From \eqref{eq:37} follows that  $R_{+-}$ and $R_{-+}$ are invertible both on the
left and right, which is possible  only if they  are square matrices,
meaning 
\begin{equation}
\text{dim} \,\HH_+ = \text{dim} \,\HH_-.
\label{eq:89}
\end{equation}
This remains true in infinite dimension (more precisely countably
infinite since $\HH$ is separable by hypothesis), as for the minimal twist of a
manifold: the rank of the spinor bundle is
finite but $\HH_\pm$ have
infinite dimension as vector spaces. 

Another useful criteria to understand whether a minimal twist may be
expandable follows from the computation of the second equation
\eqref{eq:112} for $R$ as in \eqref{eq:37}, with
$R_{++}=R_{--}=0$. Denoting $\pi_\pm$ the restrictions of $\pi$ to
$\HH_\pm$   (as in \eqref{eq:24}), one gets that $[R,\pi(a)]=0$
if, and only if  $R_{+-}\pi_-(a)=\pi_+(a)R_{+-}$ and
$R_{-+}\pi_+(a)=\pi_-(a)R_{-+}$.  In other terms, the two
representations $\pi_\pm$ of $\A$ induced by the twisting operator  must be
conjugate to each another by the action of $R_{+-}$ and $R_{-+}$:
\begin{equation}
  \pi_+(a)=R_{+-}\, \pi_-(a)\, R_{+-}^{-1}\qquad \pi_-(a)=R_{-+}\pi_+(a)\, R_{-+}^{-1.}
  \end{equation}
A easy way to check if this might be true is by taking the trace of
the equation above. If $[R,
\pi(a]=0$, then 
\begin{equation}
  \label{eq:120}
\text{Tr}\, \left(\pi_+(a)\right) = \text{Tr}\,
\left(\pi_-(a)\right) \quad\forall a\in\A.
\end{equation}
Applied to \eqref{eq:34}, this yields
\begin{equation}
  \label{eq:121}
  \text{Tr}\, \pi'(a') = \  \text{Tr}\, \pi'(\rho(a')) \quad\forall
  a'\in\A\otimes \C^2.
\end{equation}
\begin{exem}
Let $\A=\C$ act on  $\HH= \C^3$ as diagonal
  matrix $\pi(z) =\text{diag}\, (z, z, z)$.
The operator $T=\text{diag}\, (1, -1, -1)$ splits $\HH$ in
$\C\oplus\C^2$so that
 \eqref{eq:89} does not hold. In addition,   $p_+\pi(z)=\text{diag}\,
 (z, 0, 0)$, $p_-\pi(z)=\text{diag}\,
 (0, z, z)$ so $\text{Tr} \:\pi_+(z)=
 \text{Tr} \: p_+\pi(z)=z$ whereas $\text{Tr} \:\pi_-(z)=
 \text{Tr} \: p_-\pi(z)=2z$, hence \eqref{eq:120} does not hold. The same is true for
 \eqref{eq:121}: one has $\pi'((z_1, z_2))=\text{diag}\, (z_1, z_2,
 z_2)$ so that
\begin{equation}
\text{Tr}\;\pi'\!((z_1,  z_2)) = z_1 +2z_2 \quad\text{ while }\quad 
 \text{Tr}\;\pi'\!((z_2, z_1))= z_2 + 2z_1. 
\end{equation}

An example where the eigenspaces of $T$ have same dimension but
$\pi_+$ and $\pi_-$  are not
conjugate to each other  is
$\A=\C\oplus M_2(\C)$ acting block diagonally on  $\HH=\C^{10}$  as
\begin{equation}
  \label{eq:122}
  \pi(z,m)=\text{diag}\, (m,c,c,c,m,m,c).
\end{equation}
For $T=\text{diag}\, (\I_2,1,1,1,-\I_2,-\I_2,-1)$, one as
$\text{dim}\,\HH_+=\text{dim}\,\HH_-=5$ and  the representations 
\begin{equation}
  \label{eq:123}
  \pi_+(c,m)=\text{diag}\, (m, c ,c ,c,0_2,0_2,0), \qquad  
  \pi_-(c,m)=\text{diag}\, (0_2, 0, 0, 0,m, m,c)
\end{equation}
are both faithful. However, 
  $\text{Tr}  \left(\pi_+(c,m)\right)=\text{Tr} \, m + 3c$ and  
  $\text{Tr} \left(\pi_-(c,m)\right)= 2\text{Tr} \,m + c$
are not equal for every $m\in M_2(\C), c\in\C$. As well, one has 
\begin{equation}
  \label{eq:125}
  \pi'((m_1, c_1), (m_2, c_2))= \text{diag}\left(m_1, c_1, c_1, c_1, m_2, m_2, c_2\right),
\end{equation}
whose trace   $\text{Tr} \: m_1  + 2\text{Tr}\:  m_2 + 3c_1 + c_2$ is
not equal to 
\begin{equation}
\text{Tr} \;\pi'\!\left(\rho((m_1, c_1), (m_2, c_2))\right)=\text{Tr} \;\pi'\!\left((m_2, c_2), (m_1, c_1)\right)=\text{Tr}
\: m_2  + 2\text{Tr} \: m_1 + 3c_2 + c_1
\label{eq:20}
\end{equation}
for every $c, m$.
 \end{exem}
\newpage

Proposition \ref{prop:Krein} allows to show that the space $\HH$ of an
expandable  minimal twist, equipped with a twisted product, is never  an Hilbert space.
\begin{prop}
\label{lem:indef}
 All the Hermitian twisted products associated with an expandable minimal twist of a spectral
triple $(\A, \HH, D)$ by some twisting operator $T$ are indefinite and
non degenerate.
\end{prop}
\begin{proof}
 By propositions \ref{prop:prodherm} and \ref{prop:Krein},  the twisted products associated with such a minimal twist are Hermitian
if and only if  the operator $R$ in \eqref{eq:37} is of the form 
\begin{equation}
  \label{eq:91}
  R=
  \begin{pmatrix}
    0 & \cal R \\ {\cal R}^\dag &0
  \end{pmatrix}
\end{equation}
for some invertible operator $\cal R: \HH_-\to \HH_+$.
 To any  $\varphi_\pm\in
\HH_\pm$ let us associate the two elements of $\HH$
\begin{equation}
\psi=
\begin{pmatrix}
  \varphi_+ \\\varphi_-
\end{pmatrix}\: , \qquad  \tilde \psi=
\begin{pmatrix}
  \varphi_+ \\-\varphi_-
\end{pmatrix} .\label{eq:82}
\end{equation}
Denoting $\langle \cdot, \cdot\rangle_\pm$ the restriction of the inner
product $\langle \cdot, \cdot\rangle$ of $\HH$ to $\HH_\pm$, one has
\begin{align}
  \label{eq:92}
  &(\psi, \psi)_R = \langle \psi, R\psi\rangle = \langle\varphi_+, {\cal R}
  \varphi_-\rangle_++\langle\varphi_-, {\cal R}^\dag
  \varphi_+\rangle_-,\\
  &(\tilde \psi, \tilde \psi)_R = \langle \tilde \psi, R\tilde \psi\rangle = -\langle\varphi_+{\cal R}
  \varphi_-\rangle_+-\langle\varphi_- {\cal R}^\dag
  \varphi_+\rangle_- =- (\psi, \psi)_R.
\end{align}
Since  $(\psi, \psi)_R $ is real (being the twisted product Hermitian
by hypothesis) and cannot be zero for all~$\psi$
(otherwise $\cal R$ would be zero), there exist $\varphi_\pm$ for
which $(\psi, \psi)_R$ and $(\tilde\psi, \tilde\psi)_R$ have opposite
sign. Hence the twisted product is indefinite.
\end{proof}

If the operator $R$ can be chosen with pure point spectrum and without
$0$ as an accumulation point, then by proposition
\ref{prop:Kreinprodcut} the space $\HH$ equipped with the twisted
product $(\cdot,  \cdot )_R$ is a Krein space. In infinite dimension,
it is  not clear to the author whether such an $R$ always exists. It
exists however in the examples relevant for the Standard Model. 

\begin{exem}
 As explained at the beginning of this section, for the minimal twist
 of a manifold a natural choice of operator $R$ is the
 $\gamma^0$ matrix. Actually, any Dirac matrix satisfies the conditions of proposition
 \ref{prop:Krein} and can be chosen to expand the twisting
 automorphism. It also satisfies the condition of proposition \ref{prop:Kreinprodcut}
 so the twisted product is a Krein product. Furthermore, these
 matrices (in the chiral representation) are unitary, so they are a fundamental
 symmetry. However they are not all equivalent with respect to the
 fermionic action: as explained in \cite{Martinetti:2025ab}, only $\gamma^0$ allows to interpret the integrand in the
  fermionic action as a  Lorentzian lagrangian \cite{Martinetti:2019aa}. 

In the spectral triple of electrodynamics \cite{Dungen:2011fk}, $\A=\cinf\otimes \C^2$ acts on  $\HH=L^2(\M,S)\otimes \C^4$ as
\begin{equation}
  \label{eq:128}
  \pi((f,g))= \text{diag}\, (f\,\I_4,\, f\,\I_4,  g\,\I_4, g\,\I_4) 
\end{equation}
where the $\I_4$ 's act on $L^2(\M, S)$. Its minimal twist by the
grading  
\begin{equation}
\gamma^5\otimes
\gamma_F =
\begin{pmatrix}
  \I_2 &0 \\ 0&-\I_2 
\end{pmatrix}\otimes
\begin{pmatrix}
  1&0&0&0\\ 0& -1 &0&0\\ 0&0&-1&0 \\ 0&0&0&1
\end{pmatrix}
\end{equation}
is built upon the representation of $\A\otimes \C^2$ \cite[eq. 5.8]{Martinetti:2019aa}
\begin{equation}
  \label{eq:129}
  \pi'( (f, g,f',g'))=\text{diag}\, (f\I_2,\, f'\I_2, \, f'\I_2,\,
  f\I_2,\, g'\I_2,\, g\I_2,\,  g\I_2,\,  g'\I_2,). 
\end{equation}
 The operator   
\begin{equation}
  \label{eq:127}
  R=\gamma^0\otimes \I_4
\end{equation}
induces the flip automorphism, satisfies the condition of
propositions \ref{prop:Krein}, \ref{prop:Kreinprodcut} and
is unitary. Thus its induces a Krein product on $\HH$ and is a
fundamental symmetry. The same is true for $\gamma^a\otimes \I_4$
($a=1,2,3$) as well, but only  the operator \eqref{eq:127} allows to obtain Dirac
equation in Lorentzian signature in the fermionic action \cite[\S 5]{Martinetti:2019aa}.

Regarding the Standard Model, as explained in \S\ref{subsec:minsm}, in
order to generate the extra scalar field $\sigma$  
the minimal twist should not be
done by the grading but by the operator 
\begin{equation}
\gamma_\M\otimes T_F\label{eq:130}
\end{equation}
with
$T_F$ defined in \eqref{eq:27} \cite[Remark 3.1]{Filaci:2021ab}.  The twisted
product and the fermionic action have not been investigated yet. By
proposition \ref{prop:Krein}   the twisted Standard Model
is expandable by any operator
\begin{equation}
R=\gamma^a\otimes\I_{96}
\label{eq:131}
\end{equation}
which, by Prop. \ref{prop:Kreinprodcut}  turns the
Hilbert space into a Krein space. One expects that for $a=0$, the
twisted fermionic actions will give equations of motion in Lorentzian
signature. Non published results of Manuele's PhD go into that direction.\end{exem}

\subsection{Twisted unitaries}
\label{subsec:unit}

Given an expandable  twisted spectral triple $(\A, \HH, D), \rho$ with
twisted product $(\cdot, \cdot )_R$, an operator
$u\in{\cal B}(\HH)$ is said \emph{twisted-unitary} if
\begin{equation}
  \label{eq:132}
  (u\varphi, u\psi)_R = (\varphi, \psi)_R\qquad\forall \varphi,
  \psi\in \HH.
\end{equation}
For the minimal twist of a manifold of dimension $n=2m$ with twisted
product \eqref{eq:38}, a twisted unitary $\tilde u$ whose action on the individual
components $\psi_n$ is trivial (for $n=1, ..., p$ with $p=2^m$ the dimension of the
spin representation)  is a square matrix of dimension $p$
which preserves the product
\begin{equation}
  \label{eq:51}
  \langle \cdot,\,  \gamma^0 \cdot\rangle_p .
\end{equation}
Since $\gamma^0$ is unitarily equivalent to $\text{diag}(\I_{\frac p2},-\I_{\frac p2})$, the group $\tilde U_R$ generated by these elements  is the
unitary group $U(\frac p2, \frac p2)$. In particular, on a four dimensional manifold, one
has $p=2^2=4$ so that 
\begin{equation}
  \label{eq:52}
  \tilde U_R =U(2, 2).
\end{equation}

This group is tightly linked to the symmetry group of twistors \cite{R.-Penrose:1973aa}. A possible
relation between the spectral description of the Standard Model and
twistors has long  been pointed out \cite[\S
9.1]{Connes:2010fk}. The discussion above suggests that twistors may be
relevant for a Lorentzian version of noncommutative geometry,
especially in the light of recent works on Wick rotation for
spinors \cite{Woit:2023aa}, that show some similarities with the
change of signature presented here.

\section{Conclusion}
In the last years, Manuele has worked intensively  on the computation
of the spectral action on Lorentzian manifold, aiming ad adapting some
results of \cite{Wrochna:2020aa} using technics of quantum field
theory on curved spacetime.  He has not be given enough time to finish
this work,  but made us promise to complete it.


 \bibliographystyle{abbrv}
 \bibliography{/Users/pierre/physique/articles/biblio} 
\end{document}